\documentclass[aps,preprintnumbers,nofootinbib]{revtex4}

\usepackage{epsfig,graphics}
\usepackage{amsfonts,amssymb,amsmath} 
\usepackage{amsthm}

\newcommand{\comment}[1]{}
\newcommand{\ket}[1]{\left |  #1 \right\rangle}

\theoremstyle{plain}
\newtheorem{theorem}{Theorem}
\newtheorem{lemma}{Lemma}

\theoremstyle{definition}
\newtheorem{definition}{Definition}

\begin{document}

\title{Hidden variable models for quantum theory cannot have any local part}

\author{Roger Colbeck and Renato Renner \\[-2ex] \mbox{}}

\affiliation{Institute for Theoretical Physics \\ ETH Zurich, Switzerland \\ \vspace{-1.5ex} \\ {\tt colbeck@phys.ethz.ch $\quad$ renner@phys.ethz.ch}}

\date[]{}

\begin{abstract}
  It was shown by Bell that no local hidden variable model is
  compatible with quantum mechanics.  If, instead, one permits the
  hidden variables to be entirely non-local, then any quantum
  mechanical predictions can be recovered.  In this paper, we consider
  general hidden variable models which can have both local and
  non-local parts.  We then show the existence of (experimentally
  verifiable) quantum correlations that are incompatible with any
  hidden variable model having a non-trivial local part, such as the
  model proposed by Leggett.
\end{abstract}

\maketitle

\section{Introduction}

\label{sec:HVMs}
Consider a source emitting two particles, which travel to two
detectors, located far apart.  The detectors are controlled by Alice
and Bob.  We denote Alice's choice of measurement by $A$, and
similarly Bob's by $B$.  The measurement devices generate the outcomes
$X$ and $Y$ on Alice's and Bob's sides, respectively.
\comment{This setup is shown in Figure \ref{fig:setup}.}

In a hidden variable model, one attempts to describe the outcomes of
such measurements by assuming that there are hidden random variables,
in the following denoted by $U$, $V$, and $W$, distributed according
to some joint probability distribution $P_{U V W}$. (For reasons to be
clarified below, we consider three different hidden variables.)
Measurement outcomes then only depend on Alice and Bob's choice of
measurement $A$ and $B$ as well as the values of the hidden variables
$U,V, W$, that is, formally
\begin{align*}
X & = f(A,B,U,V,W) \\
Y & = g(A,B,U,V,W) 
\end{align*}
for some functions $f$ and $g$. Rephrased in the language of
conditional probability distributions, these conditions read
\begin{align*}
P_{X|A=a,B=b,U=u,V=v,W=w}(x) & = \delta_{x, f(a,b,u,v,w)} \\
P_{Y|A=a,B=b,U=u,V=v,W=w}(y) & = \delta_{y, g(a,b,u,v,w)} \ .
\end{align*}

In this work, we divide the hidden variables into local and non-local
parts:\footnote{\label{ft:dist}Notice that our definition of
 \emph{local} and \emph{non-local} parts is not the same as that used
 in~\cite{BKP}. While ours is based on a distinction between local
 and non-local hidden variables, the definition in~\cite{BKP} relies
 on a convex decomposition of the conditional probability
 distribution into a local conditional distribution and a non-local
 one.} $U$ and $V$ are, respectively, Alice's and Bob's \emph{local}
hidden variables, and $W$ is a \emph{non-local} hidden variable. The
requirement is that, when the non-local part $W$ is ignored, Alice's
distribution depends only on the local parameters $A, U$ and Bob's
only on $B, V$,
\begin{align}  \label{eqn:loc1}
\sum_w P_W(w)P_{X|A=a,B=b,U=u,V=v,W=w} & \equiv P_{X|A=a,B=b,U=u,V=v} \equiv P_{X|A=a,U=u}
\\
\label{eqn:loc2}
\sum_w P_W(w)P_{Y|A=a,B=b,U=u,V=v,W=w} & \equiv P_{Y|A=a,B=b,U=u,V=v} \equiv P_{Y|B=b,V=v}.
\end{align}

We stress here that identities~\eqref{eqn:loc1} and~\eqref{eqn:loc2}
do not restrict the generality of the model; they are merely a
definition of what we call \emph{local}. In fact, any possible
dependence of the individual measurement outcomes $X$ and $Y$ on the
choice of measurements $A$ and $B$|in particular, the predictions of
quantum theory|can be recreated by an appropriate choice of functions
$f$ and $g$ that depend on the non-local variable $W$ but not on the
local variables $U$ and $V$. In the following, we call such a model
\emph{entirely non-local}. The \emph{de Broglie-Bohm theory} (see,
e.g., \cite{Bell}) is an example of such a model.

In the \emph{Bell model} \cite{Bell}, one makes the assumption that
the individual measurement outcomes are fully determined by local
parameters, i.e., that the functions $f$ and $g$ only depend on the
local variables $U$ and $V$, respectively, but not on $W$. It is well
known that such an assumption is inconsistent with quantum theory.
Modulo a few loopholes (see for example
\cite{Aspect,BCHKP,Kent_loophole} for discussions), experiment agrees
with the predictions of quantum mechanics, and hence falsifies Bell's
model.

\emph{Leggett} \cite{leggett} has introduced a hidden variable model
for which the hidden variables have both a local and a global part as
above.  In addition, he assumes that the expectation values of the
measurement outcomes obey a specific law (Malus' law), which depends
only on local quantities. More concretely, the assumption is that the
measurement outcomes $X$ and $Y$ are binary values and that the
measurement choices $A = \vec{A}$ and $B = \vec{B}$ as well as the
local hidden variables $U = \vec{U}$ and $V = \vec{V}$ are unit
vectors. The conditional probability distributions $P_{X|A=a,U=u}$ and
$P_{Y|B=b,V=v}$ on the r.h.s.\ of~\eqref{eqn:loc1}
and~\eqref{eqn:loc2} are given by
$[|\vec{A}\cdot\vec{U}|^2,1-|\vec{A}\cdot\vec{U}|^2]$ and
$[|\vec{B}\cdot\vec{V}|^2,1-|\vec{B}\cdot\vec{V}|^2]$, respectively.
Such a model is inconsistent with quantum theory, and has motivated
recent experiments \cite{GPKBZAZ,BLGKLS,PFGJZAZ}.

In this paper, we show that there exist quantum correlations for which
\emph{all} hidden variable models with a non-trivial local part are
inconsistent.  More precisely, we show that there is a Bell-type
experiment with binary outcomes\footnote{Note that the labeling of the
outcomes $X$ and $Y$ is irrelevant for the argument. For
concreteness, one might think of $X \in \{-1, 1\}$ or $X \in
\{0,1\}$.}  $X$ and $Y$ such that $P_{X|A=a,U=u}=\mathcal{U}_X$ for
all $a$ and $u$, and $P_{Y|B=b,V=v}=\mathcal{U}_Y$ for all $b$ and $v$
are the only distributions compatible with quantum mechanics, where
$\mathcal{U}_X$ and $\mathcal{U}_Y$ denote the uniform distributions
on $X$ and $Y$, respectively.  In particular, $X$ and $Y$ are
independent of the local hidden variables $U$ and $V$. Thus, the only
hidden variable model compatible with quantum mechanical predictions
is entirely non-local. This is in agreement with a similar result
obtained independently by Branciard \emph{et al.}~\cite{BBGKLLS},
where it is shown that Leggett-type inequalities have no local part.

\section{Definitions and Useful Lemmas}

Our technical theorem will rely on the notion of non-signaling
distributions. Intuitively, a conditional distribution $P_{X Y | A B}$
is non-signaling if the behavior on Bob's side, specified by $B$ and
$Y$, cannot be influenced by Alice's choice of $A$, and vice versa. We
give a general definition for $n$ parties.

\begin{definition}
An $n$ party conditional probability distribution $P_{X_1,\ldots
 ,X_n|A_1,\ldots ,A_n}$ is {\it non-signaling} if, for all subsets
$S\subseteq\{1,\ldots,n\}$, we have
\begin{equation*}
P_{X_{S_1},\ldots ,X_{S_{|S|}}|A_1,\ldots
 ,A_n}=P_{X_{S_1},\ldots ,X_{S_{|S|}}|A_{S_1},\ldots ,A_{S_{|S|}}}.
\end{equation*}
\end{definition}

In the following, we denote by $D(P_X,Q_X)$ the \emph{statistical
distance} between two probability distributions $P_X$ and $Q_X$,
defined by $D(P_X,Q_X) = \frac{1}{2}\sum_x|P_X(x)-Q_X(x)|$. It is easy
to verify that 
\begin{equation} \label{eq:distmax}
D(P_X, Q_X) = \sum_{x} \max[ 0, Q_X(x) - P_X(x) ] \ .
\end{equation}
Furthermore, taking marginals cannot increase the statistical
distance, i.e.,
\begin{equation}  \label{eq:dist_dec}
D(P_X,Q_X) \leq D(P_{X Z},Q_{X Z}) \ ,
\end{equation}
where $P_X$ and $Q_X$ are the marginals of joint distributions $P_{X
Z}$ and $Q_{X Z}$, respectively.  In fact, if the marginals $P_Z$
and $Q_Z$ are equal, then the distance $D(P_{X Z},Q_{X Z})$ can be
written as the expectation of the distance between the
corresponding conditional probability distributions,
\begin{equation} \label{eq:averagedist}
D(P_{X Z},Q_{X Z})
=
\sum_{z} P_Z(z) D(P_{X | Z=z} Q_{X | Z=z}) \ .
\end{equation}

Finally, we will use the following lemma which relates the statistical
distance to the probability that two random variables take the same
value.

\begin{lemma}
\label{lem:dist}
Given a joint probability distribution $P_{X Y}$, the distance
between the marginals $P_X$ and $P_Y$ is upper bounded by the
probability that $X \neq Y$, that is, 
$D(P_X,P_Y)\leq\sum_{x \neq y} P_{X Y}(x,y)$.
\end{lemma}

\begin{proof}
Define $X'$ as a copy of $X$, so that
$P_{X X'}(x,y)=0$ for all $x \neq y$.
Using~\eqref{eq:dist_dec} and~\eqref{eq:distmax}, we
have
\begin{equation}
D(P_X,P_Y)\leq D(P_{X X'},P_{X Y})=\sum_{x \neq y} P_{X Y}(x,y).
\end{equation}
\end{proof}

\section{Chained Bell Inequalities}
\label{sec:measurements}

We use the family of Bell inequalities introduced by Pearle
\cite{pearle} and Braunstein and Caves \cite{BC}.  Each member of this
family is indexed by $N \in \mathbb{N}$, the number of measurement
choices.  Alice can choose the measurements $A\in\{0,2,\ldots,2N-2\}$
and Bob $B\in\{1,3,\ldots,2N-1\}$.  Each measurement has two outcomes,
i.e., $X$ and $Y$ are binary.  If $x$ is one outcome, $\bar{x}$
denotes the other. 

The quantity we consider is
\begin{equation} \label{eq:Idef}
I_N \equiv I_N(P_{X Y | A B}) := \sum_{\genfrac{}{}{0pt}{}{a,b}{|a-b|=1}}\sum_x P_{X Y|A=a, B=b}(x,\bar{x})+\sum_x P_{X Y|A=0, B=2 N-1}(x,x) \ .
\end{equation}
Note that, for any fixed $a,b$, the sum $\sum_x P_{X Y|A=a,
 B=b}(x,\bar{x})$ corresponds to the probability that the values $X$
and $Y$ are distinct.  It is easy to verify (but we are not going to
use this fact) that all classical correlations satisfy $I_N\geq 1$
(i.e., $I_N \geq 1$ is a Bell inequality), and that the CHSH
inequality~\cite{CHSH} is the $N=2$ version. We also emphasize that
the bound $I_N \geq 1$ is independent of the actual measurements
chosen and hence allows a device-independent falsification of hidden
variable models (in contrast to Leggett-type inequalities).

Using a quantum mechanical setup, one can obtain a value of $I_N$,
denoted $I_N^{\text{QM}}$, which is arbitrarily small in the large $N$
limit.  To see this, suppose Alice and Bob share the state
${1}/{\sqrt{2}}\left(\ket{00}+\ket{11}\right)$, and their
measurements take the form of projections onto the states
$\cos\frac{\theta_i}{2}\ket{0}+\sin\frac{\theta_i}{2}\ket{1}$ and
$\sin\frac{\theta_i}{2}\ket{0}-\cos\frac{\theta_i}{2}\ket{1}$, where
$\theta_i=\frac{i\pi}{2 N}$ (Alice's measurements take $i=a$, and
Bob's take $i=b$).  Using this setup, the probability that Alice and
Bob's measurement outcomes $X$ and $Y$ are distinct, for $|a-b|=1$, is given by
\[
\sum_x P_{X Y|A=a, B=b}(x,\bar{x})=\sin^2\frac{\pi}{4 N}
\]
and, likewise, the probability that the outcomes are equal for $a=0$ and $b=2N-1$ is
\[
\sum_x P_{X Y|A=0, B=2 N-1}(x,x)=\sin^2\frac{\pi}{4 N} \ .
\]
Thus, quantum mechanics predicts
\begin{align} \label{eq:INQM}
I_N^{\text{QM}}=2 N\sin^2\frac{\pi}{4N} \ ,
\end{align}
which, in the limit of large $N$, is approximated by
$\frac{\pi^2}{8 N}$ and can be made arbitrarily small.

\section{Technical Result}

Our argument is based on a straightforward extension of a result about
non-signaling distributions $P_{X Y | A B}$ by Barrett, Kent, and
Pironio~\cite{BKP}. The main difference between their result and our
Theorem~\ref{thm:main} is that our statement holds with respect to an
additional third party with an input $C$ and output $Z$. (When
applying the theorem, the local hidden variables will take the place
of $Z$, whereas $C$ is not used.)

For the following, let $X$ and $Y$ be binary, $A \in \{0, 2, \ldots,
2 N-2\}$, $B \in \{1, 3, \ldots, 2 N-1\}$ for some $N \in \mathbb{N}$, as
in Section~\ref{sec:measurements}, and let $Z$ and $C$ be
arbitrary.

\begin{theorem}
\label{thm:main}
If $P_{X Y Z | A B C}$ is non-signaling then, for any $C$ chosen
independently of the inputs $A$ and $B$,\footnote{Because $C$ is
chosen independently of $A$ and $B$, the joint distribution $P_{X Y
  Z C | A B}$ is given by $P_{X Y Z C | A=a, B=b}(x,y,z,c) = P_{X Y
  Z | A=a, B=b, C=c}(x,y,z) P_C(c)$.}
\begin{equation*}
D(P_{X Z C|A=a},\mathcal{U}_X\times
P_{Z C})\leq\frac{I_N}{2}\qquad\text{and}\qquad
D(P_{Y Z C|B=b},\mathcal{U}_Y\times P_{Z C})\leq\frac{I_N}{2}
\end{equation*}
for all $a,b$, where $I_N \equiv I_N(P_{X
Y | A B})$, and where $\mathcal{U}_X$ and $\mathcal{U}_Y$
denote the uniform distributions on $X$ and $Y$, respectively.
\end{theorem}

\begin{proof}
Using Lemma \ref{lem:dist} and the triangle inequality, we have for
any fixed $z$ and $c$
\begin{align}
I_N(P_{X  Y | A B, C=c, Z=z})
&=
\sum_{\genfrac{}{}{0pt}{}{a,b}{|a-b|=1}} \sum_x P_{X Y|A=a, B=b, C=c, Z=z}(x,\bar{x})+\sum_x P_{X Y|A=0,B=2 N-1, C=c, Z=z}(x,x) \nonumber\\
&\geq \sum_{\genfrac{}{}{0pt}{}{a,b}{|a-b|=1}}D(P_{X|A=a, C=c, Z=z},P_{Y|B=b, C=c, Z=z})+D(P_{\bar{X}|A=0,C=c, Z=z},P_{Y|B=2 N-1,C=c, Z=z}) \label{eqn:dist}\\
&\geq D(P_{\bar{X}|A=0,C=c,Z=z},P_{X|A=0,C=c,Z=z}) \ .\nonumber
\end{align}
Then, since
\begin{eqnarray}
D(P_{\bar{X}|A=0,C=c,Z=z},P_{X|A=0,C=c,Z=z})=2 D(P_{X|A=0, C=c, Z=z},\mathcal{U}_X),
\end{eqnarray}
we obtain
\begin{equation}
D(P_{X|A=0, C=c, Z=z}, \mathcal{U}_X)  \leq \frac{1}{2} I_N(P_{X  Y | A B, C=c, Z=z}) \ .
\end{equation}
Taking the average over $z$ and $c$ (distributed according to $P_{Z C}
\equiv P_{Z | C} P_C$) on both sides of this inequality and
using~\eqref{eq:averagedist} we conclude
\[
D(P_{X Z C|A=0}, \mathcal{U}_X \times P_{Z C})  \leq \frac{I_N}{2} \ .
\]
The claim for arbitrary $a$ (rather than $a=0$) as well as the second
inequality of the theorem follow by symmetry.
\end{proof}

For our argument, we apply the theorem to the setup described in
Section~\ref{sec:measurements}, with $Z := (U,V)$ and $C$ equal to a
constant (i.e., $C$ is not used). Under the assumption that the hidden
variables $U$ and $V$ are independent of the inputs $a$ and
$b$,\footnote{This assumption simply says that, in an experiment, the
 choice of measurements $a$ and $b$ must not depend on the value of
 the local hidden variables. Of course, this is the case if the
 measurements are chosen at random.} we have $P_{Z|A=a, B=b} \equiv
P_Z$. This together with~\eqref{eqn:loc1} and~\eqref{eqn:loc2} implies
the non-signaling condition. Theorem~\ref{thm:main} thus gives
\begin{equation} \label{eq:thmconclusion}
D(P_{X U|A=a}, \mathcal{U}_X \times P_U) \leq
\frac{I_N}{2} \qquad \text{and} \qquad D(P_{Y V|B=b}, \mathcal{U}_Y \times P_V) \leq
\frac{I_N}{2}
\end{equation} 
for all $a$ and $b$. In particular, for $I_N \ll 1$, the bound implies
that the measurement outcomes $X$ and $Y$ are virtually independent of
the local hidden variables $U$ and $V$.

\section{Implications}

Before summarizing the implications of Theorem~\ref{thm:main}, we
first stress that the contribution of this work is not a technical
one.  Our aim is to establish a connection between an argument
proposed in~\cite{BKP} and recent work on hidden variable models, in
particular Leggett-type
models~\cite{leggett,GPKBZAZ,BLGKLS,PFGJZAZ,BBGKLLS}.

Suppose an experiment is performed, using the setup described in
Section~\ref{sec:measurements}, which allows us to estimate an upper
bound $I_N^*$ on the quantity $I_N \equiv I_N(P_{X Y|A B})$ defined
by~\eqref{eq:Idef}. Then, according to~\eqref{eq:thmconclusion}, the
\emph{maximum locality} of $X$, which we measure in terms of its
dependence on the local hidden variable $U$ via $D(P_{X U|A=a},
\mathcal{U}_X \times P_U)$, is bounded by ${I_N^*}/{2}$.

For example, after many (noiseless) measurements of the CHSH quantity,
$I_4$, one would eventually get an upper bound $I_4^*$ close to
$I_4^{\text{QM}}=2-\sqrt{2}$ (see Eqn.~\eqref{eq:INQM}). Hence, the
maximum locality of a hidden variable theory compatible with these
measurements is $1-1/\sqrt{2} \approx 0.3$. This bound can be brought
closer to zero by performing experiments according to the setup
described in Section~\ref{sec:measurements} with larger
$N$.\footnote{For any given practical setup, the optimal value of $N$
 which minimizes the upper bound $I_N^*$ may depend on the specific
 noise model of the measurement devices.}  Such experiments were
proposed in~\cite{BKP}.

In the limit of large $N$, quantum mechanics predicts
$I_{\infty}^{\text{QM}}=0$.  Hence, for any hidden variable model to
describe these quantum correlations, we require $P_{X
U|A=a}=\mathcal{U}_X \times P_U$, and $P_{Y V|B=b}=\mathcal{U}_Y
\times P_V$. Consequently, the outcomes $X$ and $Y$ for any fixed pair
of measurements $(a,b)$ are fully independent of the local hidden
variables $U$ and $V$. Notice that, we can reach this conclusion using
only measurements in \emph{one} plane of the Bloch sphere on each side
(where Alice's plane contains $a$ and Bob's $b$).

Finally, we discuss the implications for Leggett's model. Using our
inequality~\eqref{eq:thmconclusion} with $N$ measurements in one plane
of the Bloch sphere we conclude in the limit of large $N$ that the
model can only be consistent with the predictions of quantum mechanics
if $\vec{U}$ and $\vec{V}$ are almost orthogonal to the measurement
plane. Hence, with measurements in only one plane, we can establish
that the local hidden variables $\vec{U}$ and $\vec{V}$ play no
r\^ole. A further advantage of the inequality~ we use over those of
the Leggett-type is that our inequalities enable a device independent
falsification of any hidden variable model with non-trivial local
part. Conversely, with the usual Leggett-type inequalities, the bound
depends on the setup, and is hence less experimentally robust.

\end{document}